\documentclass[11pt]{article}

\usepackage{enumitem}
\usepackage{booktabs}
\usepackage{amsthm}
\usepackage{amssymb}

\newcommand{\head}[1]{\textnormal{\textbf{#1}}}

\newtheorem{theorem}{Theorem}[section]
\newtheorem{lemma}[theorem]{Lemma}
\newtheorem{corollary}[theorem]{Corollary}

\theoremstyle{definition}
\newtheorem{defn}{Definition}[section]

\theoremstyle{remark}

\newtheorem*{note}{Note}

\newcommand{\eat}[1]{}

\begin{document}

\title{Network Flows Under Thermal Restrictions}
\author{Samiksha Sarwari\footnote{Department of Mathematics, IIT Roorkee, Roorkee 247 667, India} \hspace{0.8in} Shrisha Rao\footnote{IIIT Bangalore, Bangalore 560 100, India} \\ {\tt 22paiuma@iitr.ac.in} \hspace{0.5in} {\tt shrao@ieee.org}}
\maketitle

\begin{abstract}

We define a \emph{thermal network}, which is a network where the flow
functionality of a node depends upon its temperature.  This model is
inspired by several types of real-life networks, and generalizes some
conventional network models wherein nodes have fixed capacities and
the problem is to maximize the flow through the network.  In a thermal
network, the temperature of a node increases as traffic moves
through it, and nodes may also cool spontaneously over time, or by
employing cooling packets.  We analyze the problems of
maximizing the flow from a source to a sink for both these cases, for
a holistic view with respect to the single-source-single-sink dynamic
flow problem in a thermal network.  We have studied certain properties 
such a thermal network exhibits, and give closed-form solutions for the 
maximum flow that can be achieved through such a network. \\

\noindent{\bf Keywords}: max flow, network problems, graph walks,
thermal networks

\end{abstract}

\section{Introduction} \label{intro}

Many systems have components that are subject to thermal degradation,
and which therefore must be managed carefully to obey temperature
constraints.  This is particularly true of
electronics~\cite{hosseini2010,zhan2008}, but large systems such as
data centers~\cite{capozzoli2014,fulpagare2015} require extensive
thermal management as well.  It is therefore essential to monitor and
control the flow of work through the nodes of such a system, in
addition to the use of special equipment and measures for cooling.

The physics of thermal management can be quite
complex~\cite{fraundorf2003,kittel2000}, but in practical systems,
heuristics are generally used.  This is true of computer systems and
networks~\cite{kumar2016,joshi2012} as well as industrial process
systems~\cite{henze}.

Modeling of congestion in networks is also a well-known problem;
besides computer networks, it is also studied in the context of
vehicular traffic~\cite{kerner2009} and air
traffic~\cite{glockner1996}.  In traffic network models, congestion
control is attempted using variable pricing and other changes to node
characteristics~\cite{kockelman2005,black2010}.  

Besides congestion caused by a surfeit of packets or other arrivals at
a network node, there can also be constraints due to a node's
time-varying capacity.  This is most preeminently seen in wireless
sensor networks where nodes are subject to varying power
levels~\cite{pantazis2007}.

Existing works on capacitated networks and flow
routing~\cite{similar-to-mine,similar-to-coolant-problem,similar-to-coolant-problem-2}
do not address these issues; there does not seem to be any
sufficiently general way to consider thermal constraints, time-varying
network characteristics such as power levels, or the like.  In this
paper, we give models and results to address this.

In this paper, we give optimal solutions to the problem of maximum
flow through the following two models of a thermal network with
capacity constraints on nodes.  (Though we speak of temperature,
the concept of a thermal network and the respective parameter of
temperature can be suitably modified to model any network of nodes
that exhibit the same characteristics; the thermodynamics of
temperature or heat are not essential to our analyses.)

In the first model, a node that reaches a critical temperature stops
functioning, and can no longer be used to transmit packets; however,
it can cool with time.  The network in this \emph{dissipating model}
thus has the property of reviving itself over time, i.e., once the
network is exhausted (because the nodes are too hot), it is possible
to give it some rest (and let the nodes cool) so that we can again
send more packets through it.  Keeping in mind our aim to send as many
packets through the network as possible, we realize that the problem
now changes to sending the maximum possible packets while minimizing
the time during which the network becomes dysfunctional, so that the
number of packets in a given duration is maximized, which is
equivalent to saying that we maximize the rate of flow through the
network.  This is a dynamic problem, as the nodes repair themselves
with time, thereby making the state of the network depend upon one more
factor, i.e., time.  We study the transient state of the network, in
which the min node-cut-sets vary with time, and move on to analyze the
network to figure out if there exists a steady state.  This means that
we try to find out whether we need the state of the network at all time
instants to obtain a maximum flow using the Ford-Fulkerson
algorithm~\cite{ford-fulkerson}, or if there exists a closed-form
solution which depends only upon the initial state of the network and
  the information about the nodes' heat dissipation.  We therefore set
out to prove that there indeed exists a steady state of the network for
which we prove that there exists a node-cut-set which is the min node-cut-set
throughout after a certain amount of time.  Using the results, we are
able to find the value of the maximum rate of flow achievable in this
network.

In the second model, the network does not have the same self-healing
properties, but we have some kind of special packets called
\emph{cooling packets} at our disposal, which can decrease the
temperature of nodes.  So, given a dysfunctional network, we can send
these cooling packets to specific critical nodes, so as to make the
network functional again.  One advantage of this model is that we have
the liberty to send these special packets only to the nodes that
need thermal repair (i.e., cooling).  But this is also what makes this
problem more challenging than the previous one, as now we need to
figure out the optimal strategy for sending these cooling packets so
as to minimize the requirement of these packets while maximizing the
flow through the network, i.e., we need to minimize the repair cost,
while maximizing the efficiency of the system.  For this, we first
find out the value of max flow of packets using as many cooling
packets as we may require (i.e., find the max flow if we thermally
repair all the nodes of the network completely).  The question them is
if we can obtain the same amount of flow, but with a smaller number of
cooling packets used, and further, what is the least number of cooling
packets needed to ensure maximum flow.  We analyze this scenario by
finding the exact nodes that need repair, and the exact minimum
possible amount of repair that will make the network work at its best.
The trick used to solve this problem is based on the fact that the min
node-cut-set determines the max flow.  So, we do not really require any
other node-cut-set to work at capacity more than the maximum possible
capacity of the minimum node-cut set.  This means that we do not need to
repair all nodes to their best capacities, nor do we need to repair
all the nodes at all.  The next step is to identify the nodes and the
minimum capacities at which they should function, and find out the
optimal routing pattern for the same.  This is done by creating a set
of walks such that if we send cooling packets via these walks, not
only is the minimum node-cut-set is revived to its maximum capacity, but
the nodes of other node-cut-sets are also revived to the extent that none
of them becomes the limiting node-cut-set.  We prove that there exists such
a set of walks, and calculate the number of cooling packets to be sent
via these walks, and the corresponding maximum flow achieved.

Table~\ref{table1} provides an insight into the paper in brief.

\begin{table}[htbp] \label{table1}
\centering
\begin{tabular}{lp{7cm}}
\toprule
\head{Network Type} & \head{Results} \\
\midrule
static network & Theorem~\ref{max-flow-min-cut} (max-flow min-cut) \\
\midrule
uniform with dissipation & 
Corollary~\ref{max flow same} (consequence of Theorem~\ref{min cut set same} 
and Theorem \ref{max-flow-min-cut}); Theorem~\ref{value of tau};
Theorem~\ref{value of flow uniform} \\
\midrule
non-uniform with dissipation & 
Theorem~\ref{after n stages} (generalization of Theorem~\ref{min cut set same});
Theorem~\ref{value of flow non uniform}; (generalization of Theorem~\ref{value of flow uniform})
\\
\midrule
non-uniform with cooling & 
Theorem~\ref{4.4}; Theorem~\ref{4.5}; Theorem~\ref{number of cooling packets};
Theorem~\ref{4.8}; Theorem~\ref{4.9}\\
\bottomrule
\end{tabular} \caption{A Summary of the Results}
\end{table}

Overall, this paper is organized as follows.  In Section ~\ref{thermal-network}, first (in Subsection \ref{terminologies}) we introduce some of the preliminaries and provide a background on which the subsequent sections are based. The system model (Subsection ~\ref{system-model}), involves a detailed description of the constraints on the Thermal Network and the results and techniques to maximize the flow of packets subject to these constraints.
In Section ~\ref{dissipating-model}, this system model with an additional characteristic that the nodes can cool themselves down with time is considered. This is a dynamic system, the analysis of which requires an in-depth analysis of the transient and steady states of the system. We discover some properties of the system which are used to determine the rate of maximum flow that can be achieved.
In Section ~\ref{cooling-system}, the system model with another special characteristic is discussed. In this case, the nodes are not attributed with the self-cooling properties, but we have dedicated cooling packets for the purpose of repair of the network. The problem is to optimize the flow of the heating packets along with optimization of the number of cooling packets used so as to be able to reduce the maintenance cost of the network.
Such networks have some special properties with respect to their minimum node-cut-sets, which determine the maximum flow due to the max-flow min-cut theorem.

\section{Thermal Network} \label{thermal-network}
\subsection{Terminology}
\label{terminologies}
We have a network $G$ (also referred to as network) with nodes $v_{i}$ and edges $v_{i}v_{j}$ (directed from $v_{i}$ to $v_{j}$). The temperature of a node, say $v_{i}$, cannot rise above a temperature (called critical temperature, $\theta _{c_{i}}$), and cannot fall below the specified base temperature ($\theta_{0_{i}}$). A node at its critical temperature cannot be traversed any more, and is called a dysfunctional node. 
The packets (heating packets, which shall be referred to as simply packets throughout the text) have the property that they heat up the nodes they traverse by a certain amount $\triangle T_{u}$. (Hence, the temperature restriction on the nodes limits the number of packets that can traverse any node. Thus, we define the capacity $c_{i}$ of node $v_{i}$ to be the maximum number of packets that can traverse $v_{i}$ before it becomes dysfunctional.)\\
If some nodes of network become dysfunctional such that there exists no path fro the packets to travel from $s$ to $t$ , the network is said to be disconnected (or the network is said to have gone dysfunctional). Technically, this means that all nodes of some or the other node-cut-set have gone dysfunctional.

\begin{defn}\label{node-cut-set}

A \emph{node-cut-set} is a set of nodes, the removal of which, disconnects the network such that $s$ and $t$ lie in two separate blocks of the disconnected network (or equivalently, separate $s$ from $t$)

\end{defn}

The problem is essentially to maximize the flow, i.e. to obtain the max flow which is the maximum possible amount of flow from $s$ to $t$ that can be achieved through the network before the network becomes dysfunctional.

Section 2 is a special case of this basic model wherein the nodes have the capacity of cooling themselves down. We shall denote this rate by $\omega$. This phenomenon will be referred to as \emph{dissipation} drawing analogy from the natural dissipation phenomenon. However, because of the base temperature constraints, a node cannot be cooled down below $\theta _{0_{i}}$.\\
Since this network is time dependent, we are interested in maximizing the rate of flow of packets (the number of packets traveling from $s$ to $t$ per unit time), which shall be denoted by $\bar{f}$. This analysis will be conducted separately on a uniform and a non-uniform network.

\begin{defn}
\label{uniform}
A \emph{uniform network} is a network in which all the nodes have identical capacities. A network which is not uniform is called a \emph{non-uniform network}.

\end{defn}

Section-3 is another special variant of the basic model wherein we have cooling packets (entities which decrease the temperature of a node upon traversal by an amount equal to $\triangle T_{d}$). This model however does not have the dissipating properties.\\
The following is a mention in brief of the famous result we shall be using throughout and related definitions:\\
\begin{defn}
\label{capacity of a set}
The capacity $C_{M}$ of a set $M$ is defined as the sum of the capacities of all the nodes of that set. That is, 
\begin{equation}
C_{M}=\sum _{v_{i} \in M} c_{v_{i}} .
\end{equation}
\end{defn}

\begin{defn}\label{min-cut}
A min node-cut-set or minimum node-cut-set is defined as the node-cut-set whose capacity is less than or equal to the capacity of any other node-cut-set, where capacity of any node-cut-set is given by Definition ~\ref{capacity of a set}.
\end{defn}

\begin{theorem} \label{max-flow-min-cut}
\textbf{Max-Flow Min-Cut Theorem} \cite{bondy-murthy} : In any network, the value of a maximum flow is equal to the capacity of a minimum cut.
\end{theorem}

Table~\ref{table2} summarizes the notation used.

\begin{table}[htbp] 
\centering
\begin{tabular}{lp{10cm}} 
\toprule
\head{Symbol} & \head{Description} \\
\midrule
$v_{i}$ & $i^{th}$ node of the network $G$. \\
$\theta_{0i}$ & Initial/base temperature of the node $v_{i}$, which is also equal to its minimum possible temperature. \\
 $\theta_{ci}$ & Critical temperature of node $v_{i}$. The node $v_{i}$ ceases to function above this temperature. \\
 $c_{i}$ & Capacity of node $v_{i}$. \\
 $s$ & Source of the flow; the packets originate from this node.\\
 $t$ & Sink of the flow.\\
 $\triangle T_{u}$ & The amount by which a packet increases the temperature of a node upon traversal.\\
 $\omega$ & The amount by which the temperature of a node decreases per unit time.\\
 $\triangle T_{d}$ &  The temperature by which a cooling packet decreases the temperature of a node upon traversal subject to conditions mentioned in Section~\ref{max using cooling packets}.\\
 $\bar{f}$ & Rate of flow of packets.\\
 $M$ & A node-cut-set of network $G$. Since a network can have many node-cut-sets, we shall refer to them as $M_{i}$ throughout the text.\\
 $W$ & The set of walks from $s$ to $t$ via the nodes of the node-cut-set for which the corresponding walk set is defined.\\
 $C_{M_{i}}$ & Capacity of the $i^{th}$ node-cut-set, which is equal to the sum of capacities of all nodes that belong to the set $M_{i}$.\\
 $\tau$ & The amount of time for which the network is given rest to dissipate heat and become functional again.\\
$\beta$ & Cooling capacity of a cooling packet.\\
\bottomrule
\end{tabular} \caption{Notation} \label{table2}
\end{table}

\subsection{System Model}
\label{system-model}
Given a network $G$, with nodes denoted by $v_{i}$, having base and critical temperatures $\theta_{0i}$ and $\theta_{ci}$, our problem is to maximize the number of packets traveling from source $s$ to the sink $t$ until the network becomes dysfunctional. The packets have the property that they increase the temperature of a node by an amount equal to $\triangle T_{u}$ units upon traversal.\\\\
\textbf{Constraints}:\\
The lower and upper limits on the temperature of the nodes imposes a constraint on the number of packets that can traverse that node. Let us denote the maximum number of packets that can traverse a node $v_{i}$ before $v_{i}$ becomes dysfunctional by $c_{i}$, the capacity of the $i^{th}$ node. Let $n$ packets be able to cross node $i$ before it gets dysfunctional. A packet increases the temperature of a node by $\triangle T_{u}$ upon traversal, which gives:\\
\begin{equation}
n \triangle T_{u} \le \theta_{ci}-\theta_{0i}
\end{equation}\begin{equation}
n \le \frac {\theta_{ci} - \theta_{0i}}{\triangle T_{u}} .
\end{equation}
Since $n$ denotes the number of packets, it has to be an integer. So, the maximum value $n$ can attain is:
\begin{equation}
n_{max} = \lfloor \frac{\theta_{ci} - \theta_{0i}}{\triangle T_{u}} \rfloor .
\end{equation}
This $n_{max}$ is in fact the capacity of node $i$ by definition. So, 
\begin{equation}
c_{i} = \Big\lfloor \frac{\theta_{ci} - \theta_{0i}}{\triangle T_{u}} \Big\rfloor .
\end{equation}
There are no such temperature restrictions on the edges. Also, the number of packets that can  be dispatched from the source or that can get into the sink at any instant do not constrain the number of packets traveling through the network. This means that as many packets as the network can allow through it at any instant can be dispatched by the source and can get absorbed into the sink.
The problem- to maximize the number of packets that can travel from source $s$ to sink $t$ through network with capacity constraints on nodes has already been solved by modifying the network(will be explained below) and applying the Ford-Fulkerson algorithm~\cite{ford-fulkerson}. Nevertheless, we mention it here in full details as it will be referred to in further analysis of more complicated networks. \\

\noindent\textbf{Node-Splitting Technique} \cite{cormen} \cite{nsdeo}: \\
This technique of node-splitting is often used for spot programming when solving flow questions having flow limitations on the nodes. Every node $v_i$ is split into two nodes $v.r_i$ ($v_i$ right) and $v.l_i$ ($v_i$ left). These two nodes are joined using a directed edge from $v.l_i$ to $v.r_i$. All edges incident into the erstwhile node $v_i$ now made incident into the node $v.l_i$ and all edges incident out of $v_i$ are made incident out of the node $v.r_i$. The directed edge joining $v.l_i$ and $v.r_i$ is given a capacity equal to the capacity of the node $v_i$. All edges of the original network are given an infinite capacity (assuming no limit exists on the flow through the edges). This transforms the node-limited flow problem to the familiar edge-limited flow problem which can be easily solved using the Ford Fulkerson or other max-flow algorithms.

\section{Dissipating Model}
\label{dissipating-model}
The base model provides an elementary yet important starting point for the analysis of much more complicated yet interesting networks, one such model being the dissipating model. A dissipating network is fundamentally the base model with nodes exhibiting certain special characteristics. These nodes have a special property of self-repair. This is done by dissipating heat with time, that is, the nodes, if given some time, lose out some of the heat, thereby cooling themselves sufficiently below the critical temperature, which makes them functional again. However, a node cannot cool itself down further beyond its base temperature.\\
The rate of dissipation will be denoted by $\omega$, i.e. a node cools down by  $\omega$ units temperature per unit time. \\\\
\begin{note} We are not discretizing time, for the sake of practicality. This means that the network can be given rest for any amount of time, not necessarily integral values. Or, equivalently, we can say that it is not necessary that a node’s temperature be reduced only by an amount which is an integral multiple of $ \omega $.
\end{note}
The problem , to maximize the rate of flow of packets from source $s$ to sink $t$ through the dissipating model, is tackled in parts, wherein the first part deals with a uniform network (Definition ~\ref{uniform}) and the second part deals with a non-uniform network.\\
\subsection{Uniform Network}
For a complete understanding of the dynamic behavior of the uniform network with dissipation, a complete analysis including both transient and steady state analysis will be performed for the problem of maximization of the rate of flow of packets in the following sub-sections. Where on one hand the transient state analysis provides insight into the dynamically changing packet flow through the network, the steady state analysis illustrates the ultimate state the network achieves, that doesn’t change with time.\\
\subsubsection{Transient State Analysis}
\label{transient uniform}
As a consequence of the assumption that the flow of packets from $s$ to $t$ requires no time, at time $\tau = 0$, no dissipation occurs while the packets travel through the network. Therefore before any dissipation occurs, maximum possible number of packets would have already had traversed the network, thereby making it dysfunctional. This stage at $\tau=0$ is then no different from the base model. Hence, the maximum flow is given by the max-flow min-cut theorem on the base model.

The following proposition speaks about what the minimum node-cut-set is going to be:\\
\begin{lemma} \label{min-cardinality}
The minimum node-cut-set in a uniform network is the one with minimum cardinality, where cardinality of a set refers to the number of nodes in the set.
\end{lemma}
\begin{proof}
By Definition~\ref{min-cut}, the min node-cut-set of network $G$ is the node-cut-set of $G$ with minimum capacity (the capacity of a set being the sum of capacities of the nodes in the set).  Therefore, here the
min node-cut-set is the set $M_{i}$, where $i$ is such that $C_{M_{i}}$ is minimized.\\
\[C_{M_{i}}= \sum _{j: v_{j} \in M_{i}} c_{j} = c \vert M_{i}\vert ,\]
where $\vert M_{i} \vert$ denotes cardinality of the set $M_{i}$, i.e. the set of nodes in the set $M_{i}$.\\
Without loss of generality, let $i=k$ for which $C_{M_{i}}$ is minimum (i.e. let $M_{k}$ be the min node-cut-set). That is:\\
\begin{equation}
\label{cmin}
C_{M_k}=\min_i C_{M_{i}} = \min_i c \vert M_{i} \vert = c \min_i \vert M_{i} \vert\end{equation}

But 
\begin{equation}
\label{cmk}
C_{M_{k}}= c \vert M_{k} \vert .
\end{equation}
Therefore, from ~(\ref{cmin}) and ~(\ref{cmk}):\\
$c \min_{i} \vert M_{i} \vert = c \vert M_{k} \vert$, which means $\vert M_{k} \vert = \min_{i} \vert M_{i} \vert$.\\
This tells us that the min node-cut-set is the one with minimum cardinality, which completes our proof. 
\end{proof}
After flow $f= c \vert M_{i} \vert$, the nodes of the node-cut-set $M_{k}$ become dysfunctional, thereby disconnecting the network ($M_{k}$ being a node-cut-set). To revive the network again, the network needs sufficient time to dissipate the heat and become functional again. Let the network be given rest for $\tau$ units of time.
\begin{note} It is to be observed that the network has again reached the same state as previously, where we had a disconnected network, which was then given $\tau$ units of rest, again a maximum possible number of packets pass through the network, and it becomes disconnected. We shall call this one cycle as one \emph{stage}. So, technically, a stage of a network is the state it goes into after it has been revived after giving some rest, and allows a certain number of packets to pass through it before becoming dysfunctional again.
\end{note}
The only parameter that might vary is $\tau$. However, it does not impact the state the network is in, as is established by the following result:\\
\begin{theorem} \label{min cut set same}  In a uniform network with dissipation, the minimum node-cut-set is going the same throughout.
\end{theorem}

\begin{proof}

We shall prove the result using induction. We shall first prove that the min node-cut-set is the same for Stage 1 and 2 (let it be $M_{k}$). We assume that the min node-cut-set is $M_{k}$ in Stage $n$. Then, if we are able to prove that in Stage $n+1$ as well, $M_{k}$ will be the min node-cut-set, we would be done.

Let us consider these stages one by one as follows:

\emph{Stage 1:} The network initially has all its nodes working at capacity c.  Then, the capacity of a node-cut-set $M_{i}$ being the sum of capacities of its nodes, we get:\\
$C^{(1)}_{M_{i}}=c \vert M_{i} \vert$\\
Min node-cut-set being the node-cut-set with minimum capacity, can be obtained by minimizing $C_{M_{i}}$ over all node-cut-sets $M_{i}$, i.e. \\
\[\min_i C^{(1)}_{M_{i}} = \min_{i} c \vert M_{i} \vert = c \min_i \vert M_{i} \vert = c \vert M_{k} \vert,\]
(where $M_{k}$ is assumed to be the node-cut-set with minimum cardinality without loss of generality).\\
So, the min node-cut-set in Stage 1 is the set $M_{k}$, which is the node-cut-set with minimum cardinality. The min node-cut-set in this Stage could also have been obtained directly by applying Lemma~\ref{min-cardinality}.\\

After a max flow \(f=c \vert M_{k} \vert\) by max-flow min-cut theorem), the network becomes disconnected. The residual capacity $R_{M_{i}}$ of a node-cut-set $M_{i}$ being the capacity left after a transfer of a certain number of packets through the network is given by:
\[R^{(1)}_{M_{i}}=C^{(1)}_{M_{i}}-f=C^{(1)}_{M_{i}}-C^{(1)}_{M_{k}} = c (\vert M_{i} \vert - \vert M_{k} \vert).\]

Let the network be given sufficient amount of rest (say for time $\tau$ units) so that the network becomes functional again. Then, this revived state is the Stage 2 of the network.

\emph{Stage 2:} The improved capacities ($C^{(2)}_{M_{i}}$) of the node-cut-sets after $\tau$ units of rest are:
\begin{eqnarray*}
C^{(2)} _{M_{i}} & = & R^{(2)} _{M_{i}} + \frac{\tau \omega}{\triangle T_{u}} \vert M_{i} \vert\\
& = & c \big(\vert M_{i} \vert - \vert M_{k} \vert \big) + \frac{\tau \omega}{\triangle T_{u}} \vert M_{i} \vert	,	 \forall i .
\end{eqnarray*}

Proceeding as in Stage 1, the min node-cut-set for this state of the network is obtained for that value of $i$ for which the capacity of the node-cut-set is minimum.
\begin{eqnarray*}
\min_{i}C^{(2)} _{M_{i}} & = & \min_i \Big(c \big(\vert M_{i} \vert - \vert M_{k} \vert \big) + \frac{\tau \omega}{\triangle T_{u}} \vert M_{i} \vert \Big)\\ 
& = & \min_{i} \Big(c \vert M_{i} \vert + \frac{\tau \omega}{\triangle T_{u}} \vert M_{i} \vert \Big) -c \vert M_{k} \vert
\end{eqnarray*}

which is minimum for $\min_i \vert M_{i} \vert$, i.e., the node-cut-set
with minimum cardinality, which is $M_{k}$ as per the assumption made
in Stage 1.

Now it is established that the min node-cut-set is the same for Stages 1
and 2. Next, assume that $M_{k}$ is the min node-cut-set for Stage $n$.\\

\emph{Stage $n$:} Assume that the min node-cut-set in this stage is
$M_{k}$.  Let $C^{(n)}_{M_{i}}$ denote the improved capacity of the node-cut-sets in this stage.

Since $M_{k}$ is the min node-cut-set, the max flow is equal to $C^{(n)}
_{M_{k}}$ (by the max-flow min-cut theorem) and
\begin{equation} 
C^{(n)} _{M_{k}} \le C^{(n)} _{M_{i}}  ,  \forall i .
\end{equation}

After this amount of flow takes place through the network, the network is disconnected with the residual capacities of the node-cut-sets being:\\
$R^{(n)}_{M_{i}}=C^{(n)}_{M_{i}}-C^{(n)}_{M_{k}} .$\\

Again, we revive this network by giving it sufficient time (let it be
$\tau$ units) to be functional again. This increases the capacities of
the nodes by $\frac{\tau \omega}{\triangle T_{u}}$ and makes
the network functional again. In this next stage of the network:

\emph{Stage n+1:} 
\begin{eqnarray*}
C^{(n+1)} _{M_{i}} & = & R^{(n)} _{M_{i}} + \frac{\tau \omega}{\triangle T_{u}} \vert M_{i} \vert \\
& = & C^{(n)}_{M_i} - C^{(n)}_{M_k} + \frac{\tau \omega}{\triangle T_{u}} \vert M_{i} \vert.
\end{eqnarray*}

The minimum node-cut-set is the one with minimum capacity (minimum of all the node-cut-sets). Minimizing $C^{n+1}_{M_{i}}$ over $i$, we get:
\begin{eqnarray*}
\min_i C^{(n+1)}_{M_i} & = & \min_i \Big(C^{(n)}_{M_i} - C^{(n)} _{M_k} + \frac{\tau \omega}{\triangle T_u} \vert M_i \vert \Big) \\
& = & \min_i \Big(C^{(n)}_{M_i} + \frac{\tau \omega}{\triangle T_u} \vert M_i \vert \Big) - C^{(n)}_{M_k} .
\end{eqnarray*}

We know that $C^{(n)} _{M_{i}}$ is minimum for $i=k$ (by the assumption in
Stage $n$) and $\vert M_{i} \vert$ is also minimum for $i=k$ (by
assumption).

So, $C^{(n)} _{M_{i}} + \frac{\tau \omega}{\triangle T_{u}}
\vert M_{i} \vert$ is minimum for $i=k$.  This means that
$C^{(n+1)}_{M_{i}}$ is minimum for $i=k$, or $M_{k}$ is the min node-cut-set for Stage $n+1$.
Hence, we have proved by induction that the min node-cut-set is the same in all the Stages in case of a uniform network with dissipation.
\end{proof}

\begin{corollary} \label{max flow same} 
The rate of maximum flow in case of a uniform network with dissipation is the same in every stage.
\end{corollary}

\begin{proof} 

Lemma~\ref{min cut set same} states that the minimum node-cut-set in case of Uniform Network is going to be the same in every stage. Let this node-cut-set be denoted by $M_{k}$. Then, by Max-Flow Min-Cut Theorem, the maximum flow $f$ should be equal to $C_{M_{k}}$.

Next, we need to show that this capacity is the same in each
stage. Consider $i^{th}$ stage of the network. Since in the
$(i-1)^{th}$ stage as well, the min node-cut-set would have been $M_{k}$ (by
Lemma~\ref{min cut set same}), the residual capacity of this node-cut-set after maximum flow would have become zero. Therefore, after $\tau$
units of rest, the improved capacity $C^{(i)}_{M_{k}} = R^{(i-1)}_{M_{k}}
+ \frac{\tau \omega}{\triangle T_{u}} \vert M_{k} \vert =
\frac{\tau \omega}{\triangle T_{u}} \vert M_{k} \vert$, which
is also going to be the maximum flow, $f$, through the network in this
stage ($M_{k}$ being the node-cut-set), i.e.,

\begin{equation}\label{max flow value}
f^{(i)}=\frac{\tau \omega}{\triangle T_{u}} \vert M_{k} \vert .
\end{equation}

It can be easily seen that the right hand side of the equation is
independent of the stage ($i$), which means that for $i^{th}$ stage (where $i$ is arbitrary), the capacity of the min node-cut-set, and hence the max flow is $\frac{\tau
  \omega}{\triangle T_{u}} \vert M_{k} \vert$.\\
So, the rate of flow in the $i^{th}$ stage, $\bar{f}^{(i)}$
becomes:
\[\bar{f}^{(i)} = \frac{f}{\tau} = \frac{\omega}{\triangle T_{u}} \vert M_{k} \vert \, (\mathrm{using} ~(\ref{max
  flow value})) .\]
It is to be noted that the right hand depends
neither on the stage, nor the value of $\tau$, and is thus a constant
quantity.  Hence, it is proved that the rate of maximum flow
is the same at every stage. \qedhere

\end{proof}

\begin{note} The corollary suggests that the rate of maximum flow
in case of a uniform network with dissipation is the same
throughout, which means that the network has reached the steady
state. However, it must be noted that the rate of this maximum flow
might not necessarily be the maximum rate of flow and jumping to this
conclusion might be wrong even though it seems intuitively
correct. So, what should be the value of flow and corresponding value
of $\tau$ such that the rate of flow is maximized is to be discussed
next.
\end{note}
\subsubsection{Steady State Analysis}

The corollary suggests that the rate of max flow through the network
remains the same throughout, which means that this state is indeed the
steady state of the network.

Let the rate of flow in the steady state be denoted by
$\bar{f}$. Then,
$$\bar{f}=c\frac{\vert M_{k} \vert}{\tau} .$$

The next problem is to find out the value of $\tau$ such that we obtain the maximum rate of flow of packets from $s$ to $t$ through the network.

\begin{note}
Intuitively, it seems that $\tau $ should not be so less that the
nodes of the node-cut-sets are not even able to cool down even $\triangle
T_{u}$ units, as then the capacity (the number of packets that can
traverse a node) will remain zero, and the network disconnected. Also,
$\tau$ should not be too large so that some of the nodes reach their
base temperature due to which they cannot cool down further and hence
the number of packets that can traverse through the network does not
rise as much as the time taken, which would eventually decrease the
average rate of flow. So, to be on a safer side, $\tau$ should be such
that the capacity of any node is increased by one, which means
\begin{equation} 
\tau \omega = \triangle T_{u} .
\end{equation}
\end{note}

The following result proves this claim thus establishing the
optimality of the solution:

\begin{theorem} \label{value of tau}

The rate of flow of packets in steady state through a network with dissipation is maximum for
\begin{equation} 
\tau =\frac{\triangle T_{u} }{\omega} .
\end{equation}
\end{theorem}

\begin{proof} Since each packet increases the temperature of a node by $\triangle T_{u}$, upon traversal, when $\tau=\frac{\triangle T_{u}}{\omega}$, the increase in capacity, $\triangle c$, of each node is given by: 
$\triangle c = \lfloor \frac{\tau \omega}{\triangle T_{u}} \rfloor$\\
Substituting $\tau=\frac{\triangle T_{u}}{\omega}$:

\[\triangle c= \Big\lfloor \frac{\frac{\triangle T_{u}}{\omega}\omega}{\triangle T_{u}}\Big\rfloor = \lfloor 1 \rfloor = 1.\]
So, the capacity of the min node-cut-set $M_{k}$ is increased by $\vert M_{k} \vert$ units.
Thus, the rate of flow in each stage is given by:

\[\bar{f}= \frac{\vert M_{k} \vert}{\tau} = \frac{\vert M_{k} \vert}{\frac{\triangle T_{u}}{\omega}}\]

\begin{equation}
\label{compare}
\therefore \bar{f}=\frac{\vert M_{k} \vert \omega}{\triangle T_{u}} .
\end{equation}

It will suffice to prove that for any other value of $\tau$, the rate of flow cannot be greater than the rate of flow for $\tau =\frac{\triangle T_{u} }{\omega}$. Any other value of $\tau$ can either be greater than $\frac{\triangle T_{u} }{\omega}$ or less than $\frac{\triangle T_{u} }{\omega}$. Let us consider both these cases one by one:\\\\
\textbf{Case 1: $\tau < \frac{\triangle T_{u} }{\omega}$} \\
For this value of $\tau$, the temperature of all nodes is decreased by $\tau \omega < \triangle T_{u}$ (by using: $\tau < \frac{\triangle T_{u} }{\omega}$). \\
But for a packet increases the temperature of a node by $\triangle T_{u}$ upon traversal. So, the increase in capacity of a node, $\triangle c$, by giving $\tau$ units of rest is:
\begin{equation} \label{eq2.4-01}
\triangle c = \Big\lfloor \frac{\tau \omega}{\triangle T_{u}} \Big\rfloor.
\end{equation}
But,
\begin{equation}\label{eq2.4-02}
\frac{\tau \omega}{\triangle T_{u}} <  \frac{\frac{\triangle T_{u}}{\omega} \omega}{\triangle T_{u}} = 1.
\end{equation}

Using (\ref{eq2.4-01}) and (\ref{eq2.4-02}),
\begin{equation}
\triangle c =0
\end{equation}

So, for $\tau$ units of rest, where $\tau$ is such that $\tau < \frac{\triangle T_{u} }{\omega}$, the capacity of nodes is not increased. This means that the capacity of no node of the node-cut-set $M_{k}$ can be increased. Hence, the network remains disconnected and the rate of flow is going to be zero in every stage.\\
\textbf{Case 2: $\tau > \frac{\triangle T_{u} }{\omega}$}\\
For this value of $\tau$, the temperature of all the nodes will be decreased by $\tau \omega$ units. But, $\tau \omega > \triangle T_{u}$ (given for this case). Since a packet increases temperature of a node by $\triangle T_{u}$ units, the increase in capacity of each node, $\triangle c$ is given by:
\begin{equation}
\label{tau1}
\triangle c = \Big\lfloor \frac{\tau \omega}{\triangle T_{u}} \Big\rfloor
\end{equation}
But,
\begin{equation}
\label{tau2}
\frac{\tau \omega}{\triangle T_{u}} >  \frac{\frac{\triangle T_{u}}{\omega} \omega}{\triangle T_{u}} = 1
\end{equation}

Using ~(\ref{tau1}) and ~(\ref{tau2}), 
\begin{equation}
\triangle c \ge 1.
\end{equation}

The capacity of the node-cut-set $M_{k}$ is increased by $\triangle C_{M_{k}}$ such that:
\begin{equation}
\label{cap of Mk}
\triangle C_{M_k} = \triangle c \vert M_{k} \vert \ge \vert M_{k} \vert .
\end{equation}
Since $M_{k}$ remains the min node-cut-set throughout (by Lemma ~\ref{min cut set same}), the max flow $f$ is given by the capacity of $M_{k}$, which is given by ~\ref{cap of Mk}. Hence, the rate of flow, $\bar{f}$ becomes:
\begin{eqnarray*}
\bar{f} & = & \frac{f}{\tau}\\
& = & \frac{\triangle C_{M_{k}}}{\tau}\\
& = & \frac{\lfloor \frac{\tau \omega}{\triangle T_{u}} \rfloor \vert M_{k} \vert}{\tau} .
\end{eqnarray*}
Suppose, if possible, that this rate is greater than the rate of flow given by ~\ref{compare} when $\tau = \frac{\triangle T_{u} }{\omega}$, i.e.
$$\frac{\lfloor \frac{\tau \omega}{\triangle T_{u}} \rfloor \vert M_{k} \vert}{\tau} > \frac{\vert M_{k} \vert \omega}{\triangle T_{u}}$$
$$\Big\lfloor \frac{\tau \omega}{\triangle T_{u}} \Big\rfloor > \frac{\tau \omega}{\triangle T_{u}}$$
which is not possible. Hence, our supposition is wrong. Therefore, the rate of flow in this case will always be less than or equal to the case when $\tau = \frac{\triangle T_{u}}{\omega}$\\\\
The results from Cases 1 and 2 prove that the rate of flow is maximum when $\tau = \frac{\triangle T_{u}}{\omega}$.
\end{proof}
\begin{theorem}\label{value of flow uniform}
The maximum rate of flow in any stage is given by
\begin{equation}
\bar{f}=\frac{\vert M_{k} \vert \omega}{\triangle T_{u}}.
\end{equation}
\end{theorem}
\begin{proof}
The maximum rate of flow in any stage= (maximum flow $f$)/(minimum value of $\tau$ for obtaining the max flow $f$). Therefore,
\[\bar{f}=\frac{ \vert M_{k} \vert}{\frac{\triangle T_{u}}{\omega}} = \frac{\vert M_{k} \vert \omega}{\triangle T_{u}} \ \mathrm{(using \, (\ref{max flow value}) \, and \, Theorem~\ref{value of tau})}. \qedhere \]

\end{proof}

\subsection{Non-Uniform Network}
A non-uniform network differs from a uniform network in that in this case the capacities of all the nodes may be different. Let $c_{i}$  denote the capacity of node $v_i$.\\
Since the network is the same, the node-cut-sets are going to be the same, denoted by $M_{i}$, $i=1,2, \ldots, m$. The capacity of the set $M_{i}$ is the sum of the capacities of all its nodes, and is denoted by $ C^{j}_{M_{i}}$ for stage $j$. However for the sake of convenience, the initial capacities of the nodes and node-cut-sets (which are also the same in Stage 1), will be used without the superscript.   Also, after max-flow has taken place through the network, the remaining capacities of the node-cut-sets are denoted by $R^{j}_{M_{i}}$ for stage $j$ and that of the node by $r^{j}_{i}$ and are called \emph{residual capacities}. \\
On similar lines as the uniform network, the analysis for the transient and the steady states is done separately as follows.

\subsubsection{Transient State Analysis}
Using a similar argument as in Section~\ref{transient uniform}, at $\tau$ = 0, no dissipation has taken place so far and hence the network in Stage 1 is the same as the base model, i.e.,\\
\emph{Stage 1:} We have a network having nodes $v_{i}$ with capacities $c_{i}$ respectively. Applying the max-flow-min-cut theorem, the maximum flow is given by: \\
$f^{(1)}= \min_{i}C_{M_{i}}= C_{M_{k_1}}$, say (without loss of generality, the min node-cut-set in stage 1 is assumed to be $M_{k_1}$)\\

After the flow has taken place, the residual capacities of the node-cut-sets would be:
$R^{(1)}_{M_{i}}=C_{M_{i}}-f = C_{M_{i}}-C_{M_{k_1}}$
After $f$ packets travel from $s$ to $t$, the network becomes disconnected (as all the nodes of $M_{k_1}$ become dysfunctional). For the network to become connected, at least one of the nodes $\in M_{k_1}$ should become functional (i.e. the capacity of at least one node in $M_{k_1}$ should be at least $\triangle T_{u}$). So, we give the network $\tau$ units of rest such that:\\
$\tau \omega=\triangle T_{u}$ (This value of $\tau$ gives the max flow as well as the max rate of flow, by Theorem~\ref{value of tau}.)\\

\emph{Stage 2:} After giving the network $\tau$ units of rest, the new capacity of the $i^{th}$ node-cut-set (denoted by $C^{(2)}_{M_{i}}$) becomes:\\
$C^{(2)}_{M_{i}}=R^{(1)}_{M_{i}} +  \vert M_{i} \vert =  C_{M_{i}}-C_{M_{k_1}} + \vert M_{i} \vert$ \\
i.e. the capacity increases by 1 unit for each vertex, thereby increasing the capacity by $\vert M_{i} \vert $ units.\\
For this stage, applying the max-flow-min-cut theorem gives the max flow $f^{(2)}$ as follows:\\
\begin{eqnarray*}
f^{(2)} & = & \min_i( C^{(2)}_{M_i}) \\
& = & \min_i ( C_{M_i}-C_{M_{k_1}} + \vert M_{i} \vert) \\
& = & \min_i( C_{M_i}+ \vert M_{i} \vert) - C_{M_{k_1}} \\
& = & C_{M_{k_2}}+\vert M_{k_2} \vert) - C_{M_{k_1}},
\end{eqnarray*}
where we have assumed the set $k_2$ to be such that $C_{M_{i}}+ \vert M_{i} \vert$ is minimized.\\\\
After flow of $f^{(2)}$ units from the network, the residual capacity of the $ i^{th}$ node-cut-set becomes:
\begin{eqnarray*}
R^{(2)}_{M_i} & = & C^{(2)}_{M_i}-f^{(2)}\\
& = & C_{M_{i}} - C_{M_{k_1}} + \vert M_{i} \vert - (\min_i ( C_{M_i}+ \vert M_i \vert) - C_{M_{k_1}})\\ 
& = & C_{M_i} - C_{M_{k_1}} + \vert M_{i} \vert - (C_{M_{k_2}} + \vert M_{k_2} \vert - C^2_{M_{k_1}})\\
& = & C_{M_{i}} + \vert M_{i} \vert - C_{M_{k_2}} - \vert M_{k_2} \vert .
\end{eqnarray*}\\
\emph{Stage 3:} Proceeding in a similar way, after giving the network $\tau$ units of rest, the new capacity of the $i^{th}$ node-cut-set(denoted by $C^{(3)}_{M_{i}})$ becomes:\\
\begin{eqnarray*}
C^{(3)}_{M_i} & = & R^{(2)}_{M_i} +  \vert M_i \vert \\
& = &  C_{M_i} + \vert M_i \vert - C_{M_{k_2}} - \vert M_{k_2} \vert + \vert M_{i} \vert \\
& = & C_{M_i} + 2\vert M_i \vert - C_{M_{k_2}} - \vert M_{k_2} \vert.
\end{eqnarray*}
For this stage, applying the max-flow-min-cut theorem gives the max flow, say $ f^{3}$ as follows:
\begin{eqnarray*}
f^{3} & = & min_{i}( C^{3}_{M_{i}})\\ 
& = & min_{i}( C_{M_{i}} + 2\vert M_{i} \vert - C_{M_{k_2}} - \vert M_{k2} \vert)\\ 
& = & min_{i}( C_{M_{i}} + 2\vert M_{i} \vert) - (C_{M_{k_2}} + \vert M_{k_2} \vert)\\
& = & (C_{M_{k_3}}+2\vert M_{k_3} \vert) - C_{M_{k_2}} - \vert M_{k_2} \vert ,
\end{eqnarray*}
where we have assumed the set $k_3$ to be such that $C_{M_{i}}+ 2\vert M_{i} \vert$ is minimized.\\\\
After flow of $ f^{3}$ units from the network, the residual capacity of the $i^{th}$ node-cut-set becomes:
\begin{eqnarray*}
R^{(3)}_{M_i} & = & C^{(3)}_{M_i}-f^{(3)} \\
& = & C_{M_i} + 2\vert M_i \vert - C_{M_{k_2}} - \vert M_{k_2} \vert - (C_{M_{k_3}} + 2\vert M_{k_3} \vert - C_{M_{k_2}}-\vert M_{k_2} \vert) \\
& = & C_{M_{i}} + 2\vert M_{i} \vert - C_{M_{k_3}} - 2\vert M_{k_3} \vert .
\end{eqnarray*}
This analysis of transient state suggests that the min node-cut-set might vary from stage to stage unlike in case of uniform network. The next question is whether there exists any steady state for this kind of network or not, for which we analyze the state of the network as the number of stages increases under the steady state analysis.
\subsubsection{Steady State Analysis}
Continuing in the same way as in the previous section, suppose we reach the $n^{th} $ stage, where $n$ is some large number. Then, we have the following result which proves that there indeed exists a steady state for a non-uniform network with dissipation.

\begin{theorem} \label{after n stages}

The min node-cut-set is the same throughout after $n$ stages, where
\begin{displaymath}    
n = \left\{ \begin{array}{ll}    
 0 & \textrm{if \(C_{M_k} \leq C_{M_i}\)}, \forall i \ \mathrm{and} \\
 \max_i \Big( \frac{C_{M_k} - C_{M_i}}{\vert M_i \vert - \vert M_k \vert} \Big) & \textrm{otherwise.}               
               \end{array} \right. 
\end{displaymath}
where $M_k$ denotes the node-cut-set with minimum cardinality.
\end{theorem}
\begin{proof} Continuing as above till the $n^{th}$ stage(where $n$ is some number), we get:\\
\emph{Stage n:} \\
In the $n^{th}$ stage, the flow will be given by:
\begin{eqnarray*}
f^{(n)} & = & \min_{i}((C_{M_{i}}- C_{M_{k_{1}}})+n|M_{i}|-|M_{k_{n-1}}|)\\
& = & \min_{i}( C_{M_{i}}+n|M_{i}|)- C_{M_{k_{1}}}-(n-1)|M_{k_{n-1}}| .
\end{eqnarray*}
Now we claim that as $n$ becomes large, the minimum is given by $\min_i (M_i)$, and $C_{M_i}$ is negligible in comparison with $ n |M_{i}|$. \\ 
To prove this claim, consider a set $M_{k}$ such that $\vert M_{k} \vert$ is less than $\vert M_{i} \vert $ for all $i$ except $k$. Also, $C_{M_{k}}$ may or may not be the least. We wish to prove that there exists an $n$ such that $ M_{k}$ is going to be the min node-cut-set for all stages after stage $n$. For that, we need to show that for all stages after $n$, $\min_{i}( C_{M_i}+n|M_i|)$ occurs at $i=k$.\\
Equivalently, we need to show the existence of $n$ such that 
\begin{equation} \label{ineq}
C_{M_k} + n \vert M_{k} \vert \leq C_{M_i} + n \vert M_{i} \vert, \forall i \neq k .
\end{equation}
\emph{Case 1:}
$C_{M_{k}} \leq C_{M_{i}}$ , $\forall i$. 
Also,  $\vert M_{k} \vert \leq \vert M_{i} \vert$. \\
combining the equations, we get: \\
$C_{M_{k}}  + n \vert M_{k} \vert \leq C_{M_{i}}  + n \vert M_{i} \vert,$  $ \forall n \ge 0$\\
which proves the result for this case.\\
\emph{Case 2:}\\
$C_{M_{k}} > C_{M_{i}} $, for some or all i.\\
Then, $C_{M_{k}}  + n \vert M_{k} \vert \leq C_{M_{i}}  + n \vert M_{i} \vert$ \\
$\Rightarrow C_{M_{k}} - C_{M_{i}} \leq n (\vert M_{i} \vert - \vert M_{k} \vert), \forall i$ \\
\begin{equation}
n \geq \frac{ C_{M_{k}} - C_{M_{i}}}{\vert M_{i} \vert - \vert M_{k} \vert } , \forall i.
\end{equation}
So, for $n \geq \max_i \Big(\frac{ C_{M_{k}} - C_{M_{i}}}{\vert M_{i} \vert - \vert M_{k} \vert } \Big)$, the minimum node-cut-set is always $M_{k}$ as it satisfies (\ref{ineq}). 
And there exists an $i$, which maximizes the RHS.  With this, we establish the existence of such a value $n$, thus proving our claim. \\
So, after a considerable time has elapsed, the min node-cut-set is the same throughout, i.e. $M_{k}$ such that $M_{k} = \min_{i}(|M_{i}|)$. \qedhere
\end{proof}

\begin{note} This model (non-uniform network with dissipation) is, in fact, a generalization of the uniform network with dissipation. The Case (1) of Theorem~\ref{after n stages} is a general case of the uniform network. We obtain the result that the min cut set is the same throughout for this case which is concurrent with the result ~\ref{min cut set same} of the uniform network.
\end{note}
\begin{theorem}\label{value of flow non uniform}The maximum rate of flow in every stage in the steady state is given by:
\begin{equation}
\bar{f}=\frac{|M_{k}|\omega}{\triangle T_{u}} .
\end{equation}
\end{theorem}
\begin{proof}
Since in all subsequent stages, the min node-cut-set is the same, i.e. $M_{k}$, using Lemma ~\ref{min cut set same} and Lemma ~\ref{value of tau}, the rate of flow, $\bar{f}$ in any subsequent stage is the same as in the previous case and is equal to
\[\bar{f}=\frac{|M_{k}|\omega}{\triangle T_{u}}. \qedhere \]
\end{proof}

\section{Non-Uniform Network With Cooling} \label{cooling-system}

The uniform networks, being a subset of non-uniform networks, do not require to be analyzed separately. So, here in this section, we consider a general network (non-uniform network) with a cooling mechanism. The model description goes as follows:

We have a network, with the nodes at their maximum capacities. The problem is the same- to send as many heating packets from $s$ to $t$ via the network as possible. The only way this model differs from the basic model is that we have some “cooling packets” for the repair and maintenance of the network.
\emph{Cooling packets} are the packets that can travel via the network such that they cool a node they traverse by an amount $\triangle T_{d}$. Also, a node which is already at its maximum capacity (meaning that the node is already operating at the lowest temperature possible) does not require any cooling (in fact, it cannot be cooled any further because of restrictions on the lower bound of the temperature for each node), so the cooling packet does not cool such a vertex, which essentially means that it does not lose its cooling capacity (see Definition ~\ref{cooling capacity}) while traversing that node.

\begin{defn}\label{cooling capacity}

\emph{Cooling capacity} ($\beta$) is defined as the amount by which the cooling packet can cool the nodes before getting exhausted. 
This means a cooling packet can cool at most $n$ nodes before getting exhausted, where $n$ is such that:
$$n\triangle T_{d} \le \beta, \mathrm{i.e., } \, n \leq \Big\lfloor \frac{\beta}{\triangle T_{d}}\Big\rfloor.$$
\eat{$n_{max}$ is the maximum number of nodes a cooling packet with capacity $\beta$ can cool before getting exhausted.}
\end{defn}
We consider identical cooling packets, i.e. all of them must be of
the same capacity $\beta$.  The cooling packets are meant only for
cooling purposes, and are distinct from regular packets whose flow 
from $s$ to $t$ is sought to be maximized. 

When the cooling capacity of a cooling packet is exhausted, it is
assumed to simply disappear from the network (the assumption is
concurrent with the assumptions on the cooling packet, viz cooling
packet shall only be used for cooling purpose s, which it fails to,
once its cooling capacity is exhausted).

Our problem is to find the dispatch pattern (of heating and cooling
packets) such that the flow (of heating packets) from source $s$ to
sink $t$ is maximized.

\subsection{Maximizing Flow Through the Network Using Cooling Packets}
\label{max using cooling packets}
Our problem is to find the dispatch pattern (of heating and cooling
packets) such that the flow (of heating packets) from source $s$ to
sink $t$ is maximized.\\ Initially, we have a network, with all the
nodes at their maximum possible capacity (since every node is at its
minimum possible temperature initially and hence maximum possible
number of heating packets can traverse that node before it reaches its
upper bound and becomes dysfunctional). Since, this model is exactly
similar to the base model, for calculating the maximum flow via this
network, we can follow the same approach as in the base model (applying Ford-Fulkerson Algorithm on the equivalent network (modified using the
node-splitting technique)).\\ Once, max flow has been achieved via
this network, it becomes disconnected. Let $M_{k}$ denote the
corresponding min node-cut-set which is the node-cut-set that has actually
disconnected the network. It should be noted, however, that there might
exist other nodes that have become dysfunctional, but do not belong
to the node-cut-set $M_{k}$. However, those nodes need not be identified,
as they will not play a decisive role in further analysis.\\

To connect this disconnected network, it is obvious that we need to
repair the nodes in the node-cut-set $M_{k}$. The only option available
to us is using the cooling packets for this purpose.\\

Also, we know that the network will become functional even if at least
one of the nodes in $M_{k}$ is repaired. However, it’ll only yield the
maximum flow (assuming no other node becomes a limiting factor (this
issue will be handled later in the analysis)), which is obviously
going to be less than the initial capacity of the set $M_{k}$. Since
we need to maximize the number of heating packets, we will have to
make all the nodes in $M_{k}$ working at their respective maximum
capacities.\\

The cooling packets will obviously have to be sent via directed
paths/walks to the target nodes in $M_{k}$.

\begin{defn}\label{walk}

A \emph{walk} is a directed path from a vertex $v_{1}$ to another vertex
$v_{2}$ such that a node may be traversed more than once, but any edge
is traversed just once. Specifically for this paper, walk is
used to refer to a directed walk from $s$ to $t$.

\end{defn}

Let $W$ denote the set of walks via which we have sent the cooling
packets to the nodes in $M_{k}$.

\begin{defn} \label{walks-set} 

Walk $W_{S}$ to a set of nodes $S$ is defined as a set of walks from
$s$ to $t$ such that the walks traverse all the nodes of the set $S$ once. The set $S$ is then said to be \emph{entirely spanned} by $W_S$. If the set of walks $W_S$ spans only some of the nodes of $S$ and not all, $S$ is said to be \emph{partially spanned} by $W_S$.

\end{defn}

So, the set of walks to the node-cut-set $M_{k}$ refers to a set of walks
which pass through all nodes of the set $M_{k}$.
Then, we have the following result.

\begin{lemma} \label{4.1} 

Let the cooling packets be sent to $M_k$, $M_k$ being the min node-cut-set of the network, via the set $W$ and let the
resulting network (with increased capacity of nodes) be denoted by
$G^{*}$. Then, the min node-cut-set of $G^{*}$ will either be $M_{k}$ again or
$M_{i}$, where $M_{i}$ is the node-cut-set partially spanned by $W$.

\end{lemma}

\begin{proof}	

For this, we prove that the min node-cut-set can never be the set $M_{i}$
such that $M_{i}$ is spanned entirely by $W$ and $i \neq k$.  Let the
increase in capacity of a node-cut-set $M_{i}$ in $G^{*}$ be denoted by
$\triangle C^{*}_{M_{i}}$. Then, if the set $M_{i}$ is spanned by $W$
entirely, the capacity of the set $M_{i}$ is increased by at least as
much as that of $M_{k}$, i.e. $ \triangle C^{*}_{M_{i}} \ge \triangle C^{*}_{M_{k}}$ \\
Also, since $M_{k}$ was the min node-cut-set, the residual
capacity ($R_{M_{k}}$) of $M_{k}$ after flow of $f$ units would have become
zero, whereas that of $M_{i}$ will be greater than or equal to zero.
So, after the cooling packets have been sent through the
network,
\begin{eqnarray*}
C^{*}_{M_{i}} & = & R_{M_{i}}+\triangle C^{*}_{M_{i}} \\
& \ge & \triangle
R_{M_k} + C^*_{M_k}, 
\end{eqnarray*}
\begin{equation} \label{eqn final}
\therefore \ C^{*}_{M_{i}} \ge C^{*}_{M_{k}}. 
\end{equation}
Therefore, in the next stage, the node-cut-set $M_i$ cannot be the min node-cut-set, where $i$ is such that $M_i$ is entirely spanned by $W$ and $i \neq k$. (Even if the equality holds in ~\ref{eqn final}, we can assume $M_{k}$ to be the min node-cut-set for the sake of preserving the generality of the result.)
Thus we have proved that the min node-cut-set can never be the set $M_{i}$
such that $M_{i}$ is spanned entirely by $W$ and \(i \neq k\). \qedhere

\end{proof}

\begin{lemma} \label{4.2}

The maximum possible flow via the network $G^*$ is $f$, where $f’$ is the
flow obtained by applying the max-flow-min-cut theorem on the initial
network $G$.

\end{lemma}

\begin{proof} 

In the initial network $G$, all the nodes are at the temperature
$\theta_{0i}$, which is the minimum possible temperature that the node
can attain. So, the capacity of each node is the maximum, and let the the max flow be $f$.
Let $G^*$ denote the network the capacity of nodes of which have been improved by employing the cooling packets. We wish to prove that in no case can the max flow through the network $G^*$ exceed $f$.
Let $M_{i}$ , $i=1,2,…,m$ be all the possible node-cut-sets in the network
$G^*$, and let $C_{M_{i}}$ denote their respective capacities (the
capacity of a node-cut-set is equal to the sum of the capacities of its
nodes).

The maximum flow $f^*$ is given by the max-flow-min-cut theorem as
$$f^*=\min_i C_{M_{i}} .$$

Let us assume, without loss of generality, that the node-cut-set $M_{k}$ is
the one with minimum capacity. Then, since all its nodes are at their
maximum possible capacities (because they have been cooled to their respective base temperatures by using cooling packets), it follows that $C_{M{k}}$ is working at its maximum capacity. This directly implies that
$$\max f^* = \max \min_i C_{M_i} = \max C_{M_k} = C_{M_k}.$$
This is equal to the value of max flow through the initial graph $G$. Hence, the result. \qedhere

\end{proof}

\begin{lemma} \label{4.3} 
  
For the network to yield the maximum flow, every node-cut-set must work at
least at the capacity $\kappa$, where $\kappa$ is given
by $\kappa=\min_i C_{M_{i}}=f$

\end{lemma}

\begin{proof}

We reason as follows.

\begin{itemize}

\item[(a)] Suppose it is not necessary. That is, there exists a node-cut-set, say, 
$M_{a}$ which works at the capacity $C’_{M_{a}}$ less than $f$. Then, by 
applying the max-flow-min-cut theorem,
$C’_{M_{a}} < f .$\\
So, the max flow in this case will be $f’ = C’_{M_a} < f$.\\
But, if we know that the maximum capacity of the set $M_{k}$ is
$C_{M_{k}}$ which can be attained by increasing the capacity of its
nodes to their maximum capacities, which is not impossible. If we
increase the capacities of all the node-cut-sets in this way, it is easy to
see that the max flow will then be $C_{M_{k}}$ only, because $M_{k}$
has the minimum capacity of all the node-cut-sets when all node-cut-sets are
working at their maximum capacities.

\item[(b)] Even if we do not increase the capacity of the node-cut-sets to their respective maximum, and only upto $\kappa$, even then the
max-flow-min-cut theorem says we can attain the flow equal to $f$. And
by Lemma ~\ref{4.2}, $f$ is in fact the maximum possible flow via this network
$G$.

\end{itemize}

So, we deduce that we can attain the maximum possible flow via $G$ if
the capacity of every node-cut-set is at least $\kappa=\min_i C_{M_{i}}$. \qedhere

\end{proof}

\begin{note}
Our objective now becomes: To send cooling packets through the network in such a way that all node-cut-sets work at at least the capacity given by $ \min_i C_{M_{i}}$.
It has already been proved why the max flow can’t exceed the value $f$.
So, it is established that we can attain max flow of $f$ via the network.
So, now we have a disconnected network, say $G’$, which we have to repair by sending cooling packets so as to make it functional again so that it yields the maximum flow.
\end{note}

\begin{theorem}\label{4.4}To attain the maximum flow $f$ via network $G’$, the cooling packets must be sent via directed walks such that the walks span the entire network.
\end{theorem}

\begin{proof}

The result follows immediately from Lemmas~\ref{4.1},~\ref{4.2},
and~\ref{4.3}. \qedhere

\end{proof}

\begin{note}
The theorem does not fix upon how many cooling packets are to be sent. We can safely send as many cooling packets as required to make the set $M_k$ work at its maximum capacity.  The number of packets required for the same is given by Theorem~\ref{number of cooling packets}.
\end{note}

But for sending cooling packets so as to span the entire network, it is
necessary that such a set of walks spanning the entire network exists. This is what we shall prove
next.
\begin{note}
By \emph{entire network}, we mean the entire functional network. This means that the nodes that cannot be traversed by heating packets should not be considered as  part of the network. Therefore, we define only those nodes to be a part of the network, that allow the flow of heating packets.
\end{note}
\begin{theorem} \label{4.5} 
There exists a set of walks that spans the entire network.
\end{theorem}

\begin{proof} 

Suppose there exists a node $v$ which the heating packets traverse,
but $\nexists$ any walk from $s$ to $t$ via that node. This is self
contradictory as the node being traversible by heating packet itself
implies that there exists path from $s$ to $t$ via $v$. A path is also
a walk. So, there exists a walk from $s$ to $t$ via $v$, which
directly implies that there exists a walk from $s$ to $v$, thereby
contradicting our supposition. Hence, there exists a set of walks that
spans the entire network. \qedhere

\end{proof}

\begin{theorem} \label{number of cooling packets}

For achieving the maximum flow in a network with cooling packets by
sending the cooling packets via the walks spanning the entire network,
we need\\
$$n \ge \sum _{i:v_{i} \in M_{k}}  \Big\lceil \frac{ c_{i} * \triangle T_{u}}{ \triangle T_{d}} \Big\rceil$$ \\ 
cooling packets per “max flow” number of heating packets.
\end{theorem}
\begin{proof} We have to bring the nodes in the set $M_{k}$ to their maximum capacities by sending cooling packets. Now, the number of cooling packets to be sent to node $v_{i} \in M_{k}$ with capacity $c_{i}$ is given by $n_{v_{i}}$ such that:
$n_{v_{i}}  \triangle T_{d} \ge  c_{i}  \triangle T_{u}$\\
so that\\
$n_{v_{i}} = \Big\lceil \frac{ c_{i}  \triangle T_{u}}{ \triangle T_{d}} \Big\rceil.$

Since we need to repair all the nodes in $M_{k}$, the total number of cooling packets to be sent per $f$ number of heating packets would be:
\[n= \sum_{i:v_{i} \in M_{k}} \Big\lceil \frac{ c_{i}  \triangle T_{u}}{ \triangle T_{d}} \Big\rceil.\]
But the set of walks spanning the set $M_k$ might not span the entire network. so, there may exist other nodes that the walks did not cover. For spanning the entire network, we need to employ more cooling packets for such unspanned nodes. This results in am increase in the number of cooling packets to be sent per $f$ heating packets, and hence,
\[n \ge \sum_{i : v_i \in M_{k}} \Big\lceil \frac{ c_{i}  \triangle T_{u}}{ \triangle T_{d}} \Big\rceil. \qedhere\]

\end{proof}

\subsection{Reducing the Number of Cooling Packets and the Cooling Capacity Required}

Do we really need to send cooling packets so as to span the entire
network? Perhaps not. It seems a little counter-intuitive, but we have
the following results to substantiate the realization.

\begin{lemma} \label{4.7}

Every walk in $W$ that spans the set $M_{k}$ traverses at least one
node of each node-cut-set $M_{i}$, \(i=1,2,\ldots,m\).

\end{lemma}

\begin{proof} 

We prove the result by contradiction.  Let us assume that there
exists a set, say $M_{a}$ and a walk $w_{i}$ such that $w_{i}$ does
not traverse any vertex of $M_{a}$. Then, if all the nodes of
$M_{a}$ become dysfunctional, there still exists a path $w_{i}$ from
$s$ to $t$, which contradicts the fact that $M_{a}$ is a node-cut-set.
Hence, every walk $w_{i} \in W$ traverses at least one vertex of
every node-cut-set. \qedhere

\end{proof}

For the next result, we need to define what we mean by a walk through a node:

\begin{defn}
 \emph{A walk via a node $v_{i}$} is a walk from $s$ to $t$ via a walk such that the node $v_{i}$ lies on that walk(or equivalently, the walk traverses the node $v_i$).
 \end{defn}
 
\begin{theorem} \label{4.8}

Maximum flow $f$ in $G^*$ can be achieved by sending cooling packets to
the nodes of $M_{k}$ via walks from $s$ to $t$ via nodes in
$M_{k}$ such that the set of walks, say W, spans the entire set
$M_{k}$.

\end{theorem}

\begin{proof} We are given a set $W$ of walks that span $M_{k}$, and via which we are sending the cooling packets. Now, let $R_{M_{i}}$ denote the residual capacity of the node-cut-sets in the network $G’$ be denoted by $R_{M_{i}}$ and let the increase in capacity of a node-cut-set $M_{i}$ be denoted by $\triangle C_{M_{i}}$. Let the resultant capacity of the node-cut-set $M_{i}$ in the network $G^{*}$ be represented by $C^{*}_{M_{i}}$.
Now, when we send cooling packets via nodes in $M_{k}$ such that the capacity of the node-cut-set is increased by f, using Lemma~\ref{4.7}, the capacity of all other node-cut-sets is increased at least by $f$, i.e., $\triangle C_{M_{i}} \ge \triangle C_{M_{k}}$.

Also, since $M_{k}$ was the min node-cut-set, $C_{M_{k}} \le C_{M_{i}}$.

And after flow $f$ has taken place, in the resultant network $G’$,
$$C_{M_{k}}-f \le C_{M_{i}}-f $$
i.e., $R_{M_{k}} \le R_{M_{i}}.$  \\

Therefore, $$R_{M_{k}} + \triangle C_{M_{k}}  \le R_{M_{i}} + \triangle C_{M_{i}}$$

i.e., $C^{*}_{M_{k}} \le C^{*}_{M_{i}}.$ \\

So, using max-flow-min-cut on the network $G^{*}$, we obtain the max flow $C^{*}_{M_{k}}$, which is equal to $f$.

Since $f$ is the maximum possible flow that can ever be achieved via $G$, we have thus obtained an improved approach to obtain the max flow through $G^{*}$. \qedhere

\end{proof}

\begin{theorem}
For achieving max flow in a cooling network, we need 
\[ n = \sum _{i:v_{i} \in M_{k}}  \Big\lceil \frac{ c_{i}  \triangle T_{u}}{ \triangle T_{d}} \Big\rceil\]
number of cooling packets per max flow number of heating packets.
\end{theorem}
\begin{proof}
The previous result implies that now we do not need to send cooling packets to span the entire network $G$, rather our objective is to span the entire set $M_{k}$, and send cooling packets so that all nodes of $M_{k}$ function at their respective maximum capacities so that the set $M_{k}$, which is going to be the min node-cut-set in the subsequent stage, works at its maximum possible capacity which yields the maximum flow.\\ 
The number of cooling packets that span $M_k$ entirely so that all nodes of $M_k$ work at their respective maximum capacities is given by Theorem ~\ref{number of cooling packets} to be:
\[ n = \sum _{i:v_{i} \in M_{k}}  \Big\lceil \frac{ c_{i}  \triangle T_{u}}{ \triangle T_{d}} \Big\rceil \qedhere\]

\end{proof}

\begin{note}
As per Theorem~\ref{4.8}, our objective is just to send cooling packets via $M_k$ such that all nodes of $M_k$ work at their respective maximum capacities. Since a cooling packet loses $\triangle T_{d}$ of its cooling capacity upon traversing a node, we can save this cooling capacity by sending these cooling packets via shortest possible walks such that they span $M_k$ and make its nodes work at their respective maximum capacities. So now, we not only reduce the number of cooling packets required but also the cooling capacity required.
\end{note}
We now give two results on the value of $\beta$, first to make the network functional and then, to make the network functional such that it yields maximum flow.
\begin{theorem}\label{4.9} 
The capacity of a cooling packet required for making a dysfunctional network functional should at least be equal to the minimum of the shortest distances between $s$ and $t$ via $v_i \in M_k$ (the min node-cut-set of the network), i.e. $\beta \ge \min_{i:v_i \in M_k} d(s,v_i,t) $.
\end{theorem}
\begin{proof}
For the network to just become functional again, we need to repair at least one vertex of $M_{k}$. To reduce the cooling capacity requirement, we would send the cooling packet to the vertex $v_{i} \in M_{k}$ such that $v_{i}$ is nearest to s. Hence, if we denote by d(s, $v_{i}$,t) the distance from $s$ to $t$ via node $v_{i}$, the minimum possible value of $\beta$ required would be:
\[\min_{i:v_{i} \in M_{k}}d(s,v_{i},t). \qedhere\]
\end{proof}

\begin{theorem}
The capacity of a cooling packet required to obtain maximum possible flow through the network should at least be equal to the maximum of the shortest distances between $s$ and $t$ via $v_i \in M_k$ (the min node-cut-set of the network), i.e. $\beta \ge \max_{i:v_i \in M_k} d(s,v_i,t)$.

\end{theorem}

\begin{proof}

By Theorem ~\ref{4.8}, to obtain maximum possible flow through the network, we need to send cooling packets via the set of walks $W$ such that $W$ entirely spans $M_k$.
Also, we need to repair all the nodes of $M_{k}$ to their respective maximum capacities, so as to be able to obtain a max flow $f$. For this, we need the capacity of the cooling packets to be such that even the node ($\in M_k$) farthest from $s$ is also traversed. For that, we need capacity to be such that
$\beta = \max_{i:v_{i} \in M_{k}}d(s,v_{i},t) $, which proves the result. \qedhere

\end{proof}

\begin{note}

We are not maximizing or minimizing over the distances from $s$ to $v_{i}$, rather over distance 
from $s$ to $t$ via $v_{i}$, because if we do not traverse from $s$ to $t$, in the subsequent stage, 
$M_{k}$ might not necessarily be the min node-cut-set.

\end{note}

\section{Conclusion}

This paper defines a thermal network, and gives results for the maximum flow 
that is achievable through a thermal network. In many networks, there are 
restrictions on the nodes, which may be repair constraints or pollution 
level constraints (in road networks) or the amount of data to be transferred 
through a node that is already stressed (in computer networks).  This work 
is a generalization to all such problems, whose systems may thus be
regarded as real-life thermal networks.  An aspect of the model we have 
discussed is that it is dynamic in nature, thereby capturing the temporal 
properties of the nodes. There are infrastructures which are to be 
maintained and used for very long durations. In such networks, we have to 
maximize the flow while maintaining nodes in a manner that does not
contribute to their breakdowns.

Also, our results give the maximum flow values through the
network under such constraints which can be used to measure the amount
of error in heuristic algorithms developed for similar problems.

This paper also opens up the scope of developing exact algorithms for
such networks using the approach we have implemented.  The models can
also be extended, e.g., to get a dissipating model with the rates of 
dissipation being different for different nodes.  Algorithmic work is
also possible, especially with real-life system data; for instance,
algorithms for coolant problems, optimizing the capacity of the coolants
used, and similar applications and implementation to practical problems.

\section*{Acknowledgment}

The work of the first author was partially supported by an INSPIRE
research fellowship from the Department of Science and Technology,
Government of India.

\bibliographystyle{siam}
\bibliography{references2,thermal}

\end{document}